\documentclass{article}

\usepackage{geometry}

\usepackage{graphicx}
\usepackage{amsmath,amsthm,amssymb,amsfonts}
\usepackage{ucs}
\usepackage{blkarray}
\usepackage[utf8x]{inputenc}
\usepackage[svgnames]{xcolor}
\usepackage[all,2cell,cmtip]{xy}
\usepackage{blkarray}
\usepackage{graphicx}
\usepackage{authblk}
\UseTwocells

\usepackage{tikz}
\usetikzlibrary{arrows,backgrounds,circuits,circuits.ee.IEC,shapes,fit,matrix,decorations.markings, positioning}
\pgfdeclarelayer{edgelayer}
\pgfdeclarelayer{nodelayer}
\pgfsetlayers{edgelayer,nodelayer,main}
\tikzstyle{main node} =[circle,fill=white!20,draw,font=\sffamily\Large\bfseries]
\tikzstyle{terminal}=[circle,fill=white!20,draw,font=\sffamily\Large\bfseries,color=purple,fill=none]
\tikzstyle{shadecircle}=[circle,fill=DodgerBlue,draw=Black]
\tikzstyle{lipid}=[-,draw=Black,line width=0.750]
\tikzstyle{empty}=[inner sep=3pt]

\newcommand{\maps}{\colon}
\newcommand{\define}[1]{{\bf{\boldmath{#1}}}}
\newcommand{\R}{\mathbb{R}}

\newcommand{\LinRel}{ \mathtt{LinRel}}

\newcommand{\DetBalMark}{\mathtt{DetBalMark}}

\theoremstyle{plain}

\newtheorem{thm}{Theorem}

\newtheorem{defn}[thm]{Definition}

\theoremstyle{remark}

\definecolor{myurlcolor}{rgb}{0.6,0,0}
\definecolor{mycitecolor}{rgb}{0,0,0.8}
\definecolor{myrefcolor}{rgb}{0,0,0.8}
\usepackage[pagebackref]{hyperref}
\hypersetup{colorlinks,
linkcolor=myrefcolor,
citecolor=mycitecolor,
urlcolor=myurlcolor}

\begin{document}



  \title{\large \bf Open Markov processes: A compositional perspective on non-equilibrium steady states in biology \\ }  

\author{ \normalsize Blake S. Pollard  \thanks{ Email: bpoll002@ucr.edu} }

\affil{\em Department of Physics and Astronomy \\
\em University of California \\
\em Riverside, CA 92521 \\ }

\date{ \normalsize \today}

\maketitle


\begin{abstract}
\vskip 0.2em \noindent
In recent work, Baez, Fong and the author introduced a framework for describing Markov processes equipped with a detailed balanced equilibrium as open systems of a certain type. These `open Markov processes' serve as the building blocks for more complicated processes. In this paper, we describe the potential application of this framework in the modeling of biological systems as open systems maintained away from equilibrium. We show that non-equilibrium steady states emerge in open systems of this type, even when the rates of the underlying process are such that a detailed balanced equilibrium is permitted. It is shown that these non-equilibrium steady states minimize a quadratic form which we call `dissipation.' In some circumstances, the dissipation is approximately equal to the rate of change of relative entropy plus a correction term. On the other hand, Prigogine's principle of minimum entropy production generally fails for non-equilibrium steady states. We use a simple model of membrane transport to illustrate these concepts.
\end{abstract}


\section{Introduction}

Life exists away from equilibrium. Left isolated, systems will tend toward thermodynamic equilibrium. Open systems can be maintained away from equilibrium via the exchange of energy and matter with the environment. In addition, biological systems typically consist of a large number of interacting parts. This paper presents a way of describing these `parts' as morphisms in a category. A category consists of a collection of \emph{objects} along with \emph{morphisms} or arrows between objects, obeying certain conditions. We consider time-homogeneous Markov processes as a general framework for modeling various biological and biochemical systems whose dynamical equations are linear. Viewed as morphisms in a category, the `open Markov processes' discussed in this paper provide a framework for describing open systems which can be combined to build larger systems. 

Intuitively, one can think of a Markov process as specifying the dynamics of a probability or `population' distribution that is spread across a finite set of states. A population distribution is a non-normalized probability distribution, see for example \cite{Kingman}. The population of a particular state can be any non-negative real number. The total population in an open Markov process is not constant in time as population can flow in and out through certain boundary states. Part of the utility of Markov processes as models of physical or biological systems stems from the flexibility in choosing the correspondence between the states of the Markov process and the actual system it is to model. For instance, the states of a Markov process could correspond to different internal states of a particular molecule or chemical species. In this case, the transition rates describe the rates at which the molecule transitions among these states. Or, the states of a Markov process could correspond to a molecule's physical location. In this case, the transition rates encode the rates at which that molecule moves from place to place. 

This paper is structured as follows. In Section \ref{sec:OMP} we give some preliminary definitions from the theory of Markov processes and explain the concept of an open Markov process. In Section \ref{sec:Membrain} we introduce a model of membrane transport as a simple example of an open Markov process. In Section \ref{sec:DetBalMark}, we introduce the category $\DetBalMark$. The objects in $\DetBalMark$ are finite sets of `states' whose elements are labeled by non-negative real numbers which we call `populations'. The morphisms in $\DetBalMark$ are Markov processes equipped with a detailed balanced equilibrium distribution as well as maps specifying input and output states. If the outputs of one process match the inputs of another process the two can be composed, yielding a new open Markov process. We refer to the union of the input and output states as the boundary of an open Markov process. 

In Section \ref{sec:dissipation}, we show that if the populations at the boundary of an open detailed balanced Markov process are held fixed, then the non-equilibrium steady states which emerge minimize a quadratic form, which we call the `dissipation,' subject to the constraint on the boundary populations. Depending on the values of the boundary populations these non-equilibrium steady states can exist arbitrarily far from the detailed balanced equilibrium of the underlying Markov process. In recent work \cite{BaezFongP}, Baez, Fong and the author construct a functor $\square \maps \DetBalMark \to \LinRel$ from the category of open detailed balanced Markov process to the category of linear relations. Applied to an open detailed balanced Markov process, this functor yields the subset of allowed steady state boundary population-flow pairs, providing an effective `black-boxing' of open detailed balanced Markov processes. In Section \ref{sec:rentropy} we show that, for fixed boundary populations, this principle of minimum dissipation approximates Prigogine's principle of minimum entropy production in the neighborhood of equilibrium plus a correction term involving only the flow of relative entropy through the boundary of the open Markov process.

\section{Open Markov processes}\label{sec:OMP}

In this section we define open Markov processes, describe the detailed balanced condition for equilibria and define non-equilibrium steady states for Markov processes.

An \define{open Markov process}, or open continuous time, discrete state Markov chain, is a triple $(V,B,H)$ where $V$ is a finite set of \define{states}, $B \subseteq V$ is the subset of \define{boundary states} and $H \maps \R^V \to \R^V$ is an \define{infinitesimal stochastic Hamiltonian}
\[ H_{ij} \geq 0, \ \ i \neq j \]
\[ \sum_i H_{ij} = 0. \] 
For each $i \in V$ the dynamical variable $p_i  \in [0,\infty), \ i \in V,$ is the \define{population} at the $i^{\text{th}}$ state. We call the resulting function $p \maps V \to [0,\infty)$ the \define{population distribution}. Populations evolve in time according to the \define{open master equation}  
\[ \frac{dp_i}{dt} = \sum_j H_{ij}p_j, \ \ i \in V-B \]
\[  p_i(t) = b_i(t), \ \ i \in B. \] 
The off-diagonal entries $H_{ij}, \ i \neq j$ are the rates at which population transitions from the $j^{\text{th}}$ to the $i^{\text{th}}$ state. A \define{steady state} distribution is a population distribution which is constant in time:
\[ \frac{dp_i}{dt} = 0 \ \ \text{for all} \ \ i \in V. \]

A \define{closed Markov process}, or continuous time, discrete state Markov chain, is an open Markov process whose boundary is empty. For a closed Markov process, the open master equation becomes the usual master equation
\[ \frac{dp}{dt} = Hp. \]
In a closed Markov process the total population is conserved:
\[ \sum_i \frac{dp_i}{dt} = \sum_{i,j} H_{ij}p_j = 0, \] 
enabling one to talk about the relative probabilities of being in particular states. A steady-state distribution in a closed Markov process is typically called an \define{equilibrium}. We say an equilibrium $q \in [0,\infty)^V $ of a Markov process is \define{detailed balanced} if
\[ H_{ij}q_j = H_{ji}q_i \ \  \text{for all} \ \ i,j \in V. \] 
An \define{open detailed balanced Markov process} is an open Markov process $(V,B,H)$ together with a detailed balanced equilibrium $q \maps V \to (0,\infty)$ on $V$. Notice that the populations of all states in a detailed balanced equilibrium are non-zero.

For a pair of distinct states $i,j \in V$, the term $H_{ij}p_j$ is the flow of population from $j$ to $i$. The \define{net flow} of population from the $j^{\text{th}}$ state to the $i^{\text{th}}$ is 
\[ J_{ij}(p) = H_{ij}p_j - H_{ji}p_i. \]
Summing the net flows into a particular state we can define the \define{net inflow} $J_i(p) \in \R$ of a particular state to be
\[ J_i (p) = \sum_j J_{ij}(p) = \sum_j H_{ij}p_j - H_{ji}p_i. \]
Since $\sum_j H_{ji}p_i = 0,$ the right side of this equation is the time derivative of the population at the $i^{\text{th}}$ state. Writing the master equation in terms of $J_{ij}(p)$ or $J_i(p)$ we have
\[ \frac{dp_i}{dt} = \sum_j J_{ij}(p) = J_i(p). \]
The net flow between each pair of states vanishes identically in a detailed balanced equilibrium $q$:
\[ J_{ij}(q) = 0. \]

The existence of a detailed balanced equilibrium is equivalent to a condition on the rates of a Markov process due known as \define{Kolmogorov's criterion} \cite{Kelly}, namely that
\[ H_{i_1 i_2}H_{i_2 i_3} \dotsm H_{i_{n-1} i_n} H_{i_n i_1} = H_{i_1 i_n}H_{i_n i_{n-1}} \dotsm H_{i_3 i_2}H_{i_2 i_1} \]
for any finite sequence of states $i_1, i_2, \dotsc ,i_n$ of any length. This condition says that the product of the rates along any cycle is equal to the product of the rates along the same cycle in the reverse direction.

A \define{non-equilibrium steady state} is a steady state in which the net flow between at least one pair of states is non-zero. Thus there could be population flowing between pairs of states, but in such a way that these flows still yield constant populations at all states. In a closed Markov process the existence of non-equilibrium steady states requires that the rates of the Markov process violate Kolmogorov's criterion. We show that open Markov processes with constant boundary populations admit non-equilibrium steady states even when the rates of the process satisfy Kolmogorov's criterion. Throughout this paper we use the term equilibrium to mean detailed balanced equilibrium.

\section{Membrane diffusion as an open Markov process}\label{sec:Membrain}

To illustrate these ideas, we consider a simple model of the diffusion of neutral particles across a membrane as an open detailed balanced Markov process with three states $V=\{A,B,C\}$, input $A$ and output $C$. The states $A$ and $C$ correspond to the each side of the membrane, while $B$ corresponds within the membrane itself, see Figure \ref{fig:membrane}.

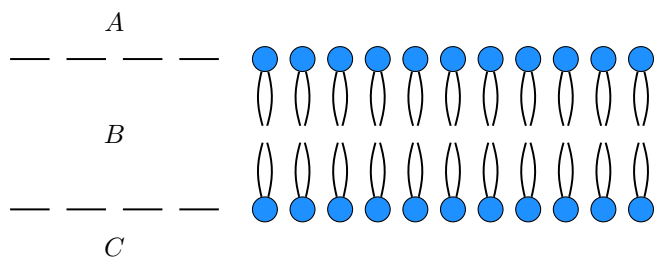
\begin{figure}[h]
\begin{center}
\begin{tikzpicture}
	\begin{pgfonlayer}{nodelayer}
		\node [style=empty] (0) at (-0.5, -0) {};
		\node [style=shadecircle] (1) at (-2, -1) {};
		\node [style=shadecircle] (2) at (2.5, -1) {};
		\node [style=empty] (3) at (0.5, -0) {};
		\node [style=shadecircle] (4) at (-2, 1) {};
		\node [style=shadecircle] (5) at (-0.5, -1) {};
		\node [style=empty] (6) at (-3, 1) {};
		\node [style=empty] (7) at (-3.75, 1) {};
		\node [style=empty] (8) at (1, -0) {};
		\node [style=shadecircle] (10) at (-1.5, -1) {};
		\node [style=shadecircle] (11) at (0.5, -1) {};
		\node [style=shadecircle] (12) at (1.5, 1) {};
		\node [style=empty] (13) at (0, -0) {};
		\node [style=empty] (15) at (-3.75, -1) {};
		\node [style=shadecircle] (16) at (0.5, 1) {};
		\node [style=shadecircle] (17) at (-2.5, 1) {};
		\node [style=empty] (18) at (2, -0) {};
		\node [style=empty] (19) at (-5.25, 1) {};
		\node [style=shadecircle] (20) at (0, 1) {};
		\node [style=empty] (21) at (-4.5, 1) {};
		\node [style=empty] (22) at (-5.25, -1) {};
		\node [style=shadecircle] (23) at (-1.5, 1) {};
		\node [style=shadecircle] (24) at (1.5, -1) {};
		\node [style=shadecircle] (25) at (2, 1) {};
		\node [style=shadecircle] (26) at (1, 1) {};
		\node [style=shadecircle] (27) at (1, -1) {};
		\node [style=empty] (28) at (-3, -1) {};
		\node [style=empty] (29) at (1.5, -0) {};
		\node [style=shadecircle] (30) at (2.5, 1) {};
		\node [style=empty] (31) at (-2.5, -0) {};
		\node [style=empty] (32) at (-4.5, -1) {};
		\node [style=empty] (33) at (-1, -0) {};
		\node [style=shadecircle] (34) at (-1, 1) {};
		\node [style=shadecircle] (35) at (-0.5, 1) {};
		\node [style=empty] (36) at (-6, 1) {};
		\node [style=empty] (37) at (2.5, -0) {};
		\node [style=shadecircle] (38) at (0, -1) {};
		\node [style=empty] (39) at (-1.5, -0) {};
		\node [style=empty] (40) at (-6, -1) {};
		\node [style=shadecircle] (41) at (2, -1) {};
		\node [style=shadecircle] (42) at (-1, -1) {};
		\node [style=shadecircle] (44) at (-2.5, -1) {};
		\node [style=empty] (45) at (-2, -0) {};
		\node [style=empty] (9) at (-4.5, 1.5) {$A$};
		\node [style=empty] (14) at (-4.5, -0) {$B$};
		\node [style=empty] (43) at (-4.5, -1.5) {$C$};
	\end{pgfonlayer}
	\begin{pgfonlayer}{edgelayer}
		\draw [style=lipid, bend right=15, looseness=1.00] (17) to (31);
		\draw [style=lipid, bend left=15, looseness=1.00] (17) to (31);
		\draw [style=lipid, bend left=15, looseness=1.00] (44) to (31);
		\draw [style=lipid, bend right=15, looseness=1.00] (44) to (31);
		\draw [style=lipid, bend right=15, looseness=1.00] (4) to (45);
		\draw [style=lipid, bend left=15, looseness=1.00] (4) to (45);
		\draw [style=lipid, bend left=15, looseness=1.00] (1) to (45);
		\draw [style=lipid, bend right=15, looseness=1.00] (1) to (45);
		\draw [style=lipid, bend right=15, looseness=1.00] (23) to (39);
		\draw [style=lipid, bend left=15, looseness=1.00] (23) to (39);
		\draw [style=lipid, bend left=15, looseness=1.00] (10) to (39);
		\draw [style=lipid, bend right=15, looseness=1.00] (10) to (39);
		\draw [style=lipid, bend right=15, looseness=1.00] (34) to (33);
		\draw [style=lipid, bend left=15, looseness=1.00] (34) to (33);
		\draw [style=lipid, bend left=15, looseness=1.00] (42) to (33);
		\draw [style=lipid, bend right=15, looseness=1.00] (42) to (33);
		\draw [style=lipid, bend right=15, looseness=1.00] (35) to (0);
		\draw [style=lipid, bend left=15, looseness=1.00] (35) to (0);
		\draw [style=lipid, bend left=15, looseness=1.00] (5) to (0);
		\draw [style=lipid, bend right=15, looseness=1.00] (5) to (0);
		\draw [style=lipid, bend right=15, looseness=1.00] (20) to (13);
		\draw [style=lipid, bend left=15, looseness=1.00] (20) to (13);
		\draw [style=lipid, bend left=15, looseness=1.00] (38) to (13);
		\draw [style=lipid, bend right=15, looseness=1.00] (38) to (13);
		\draw [style=lipid, bend right=15, looseness=1.00] (16) to (3);
		\draw [style=lipid, bend left=15, looseness=1.00] (16) to (3);
		\draw [style=lipid, bend left=15, looseness=1.00] (11) to (3);
		\draw [style=lipid, bend right=15, looseness=1.00] (11) to (3);
		\draw [style=lipid, bend right=15, looseness=1.00] (26) to (8);
		\draw [style=lipid, bend left=15, looseness=1.00] (26) to (8);
		\draw [style=lipid, bend left=15, looseness=1.00] (27) to (8);
		\draw [style=lipid, bend right=15, looseness=1.00] (27) to (8);
		\draw [style=lipid, bend right=15, looseness=1.00] (12) to (29);
		\draw [style=lipid, bend left=15, looseness=1.00] (12) to (29);
		\draw [style=lipid, bend left=15, looseness=1.00] (24) to (29);
		\draw [style=lipid, bend right=15, looseness=1.00] (24) to (29);
		\draw [style=lipid, bend right=15, looseness=1.00] (25) to (18);
		\draw [style=lipid, bend left=15, looseness=1.00] (25) to (18);
		\draw [style=lipid, bend left=15, looseness=1.00] (41) to (18);
		\draw [style=lipid, bend right=15, looseness=1.00] (41) to (18);
		\draw [style=lipid, bend right=15, looseness=1.00] (30) to (37);
		\draw [style=lipid, bend left=15, looseness=1.00] (30) to (37);
		\draw [style=lipid, bend left=15, looseness=1.00] (2) to (37);
		\draw [style=lipid, bend right=15, looseness=1.00] (2) to (37);
		\draw [style=lipid] (6) to (7);
		\draw [style=lipid] (7) to (21);
		\draw [style=lipid] (21) to (19);
		\draw [style=lipid] (19) to (36);
		\draw [style=lipid] (28) to (15);
		\draw [style=lipid] (15) to (32);
		\draw [style=lipid] (32) to (22);
		\draw [style=lipid] (22) to (40);
	\end{pgfonlayer}
\end{tikzpicture}
\caption{A simple model for passive diffusion across a membrane.}
\label{fig:membrane}
\end{center}
\end{figure}
In this model, $p_A$ is the number of particles on one side of the membrane, $p_B$ the number of particles within the membrane and $p_C$ the number of particles on the other side of the membrane. The off-diagonal entires in the Hamiltonian $H_{ij}, i \neq j$ are the rates at which population hops from $j$ to $i$. For example $H_{AB}$ is the rate at which population moves from $B$ to $A$, or from inside the membrane to the top of the membrane. Let us assume that the membrane is symmetric in the sense that the rate at which particles hop from outside of the membrane to the interior is the same on either side, i.e. $H_{BA} = H_{BC} = H_{in}$ and $H_{AB} = H_{CB} = H_{out}$. We can draw such an open Markov process as a labeled graph:
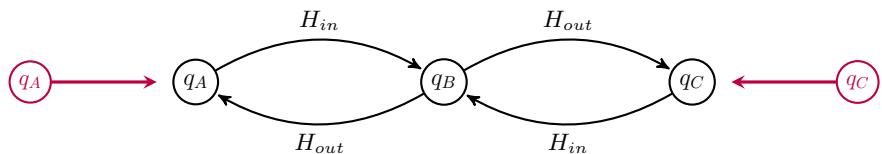
\[
\begin{tikzpicture}[->,>=stealth',shorten >=1pt,thick,scale=1.1]
\node[main node, scale=.65](A) at (-2,0) {$q_A$};
\node[main node, scale=.65](B) at (1,0) {$q_B$};
\node[main node, scale=.65](C) at (4,0) {$q_C$};
\node[terminal,scale=.6](X) at (-4,0) {$q_A$};
\node[terminal, scale=.6](Y) at (6,0) {$q_C$};

  \path[every node/.style={font=\sffamily\small}, shorten >=1pt]
    (A) edge [bend left] node[above] {$H_{in}$} (B)
    (B) edge [bend left] node[below] {$H_{out}$} (A) 
    (B) edge [bend left] node[above] {$H_{out}$} (C)
    (C) edge [bend left] node[below] {$H_{in}$} (B);

\path[color=purple, very thick, shorten >=6pt, ->, >=stealth] (X) edge (A);
\path[color=purple, very thick, shorten >=6pt, ->, >=stealth] (Y) edge (C);

\end{tikzpicture}
\]
The labels on the edges are the corresponding transition rates. The states are labeled by their detailed balanced equilibrium populations, which, up to an overall scaling, are given by $q_A = q_C = H_{in} H_{out}$ and $q_B = H_{in}^2$.
Suppose the populations $p_A$ and $p_C$ are externally maintained at constant values, i.e. whenever a particle diffuses from outside the cell into the membrane, the environment around the cell provides another particle and similarly when particles move from inside the membrane to the outside. We call $(p_A,p_C)$ the \define{boundary populations}. Given the values of $p_A$ and $p_C$, the steady state population $p_B$ compatible with these values is
\[ p_B = \frac{H_{in}p_A + H_{in}p_C}{-H_{BB} } = \frac{H_{in}}{H_{out}}\frac{p_A + p_C}{2}. \]
In Section \ref{sec:dissipation} we show that this steady state population minimizes the dissipation, subject to the constraints on $p_A$ and $p_C$.

We thus have a non-equilibrium steady state $p = (p_A, p_B, p_C)$ with $p_B$ given in terms of the boundary populations above. From these values we can compute the boundary flows, $J_A, J_C$ as 
\[ J_A = \sum_j J_{Aj}(p) = H_{out}p_B -H_{in}p_A \]
and 
\[ J_C = \sum_j J_{Cj}(p) = H_{out}p_B - H_{in}p_C. \]
Written in terms of the boundary populations this gives
\[ J_A = \frac{ H_{in}(p_C-p_A)}{2} \]
and
\[ J_C = \frac{H_{in}(p_A-p_C)}{2}. \]
Note that $J_A = - J_C$ implying that there is a constant net flow through the open Markov process. As one would expect, if $p_A > p_C$ there is a positive flow from $A$ to $C$ and vice-versa. Of course, in actual membranes there exist much more complex transport mechanisms than the simple diffusion model presented here. A number of authors have modeled more complicated transport phenomena using the framework of networked master equation systems \cite{OPKBio, SchnakenBook}.

In our framework, we call the collection of all boundary population-flows pairs the steady state `behavior' of the open Markov process. The main theorem of \cite{BaezFongP} constructs a functor from the category of open detailed balanced Markov process to the category of linear relations. Applied to an open detailed balanced Markov process, this functor yields the set of allowed steady state boundary population-flow pairs. One can imagine a situation in which only the populations and flows of boundary states are observable, thus characterizing a process in terms of its behavior. This provides an effective `black-boxing' of open detailed balanced Markov processes. 

As morphisms in a category, open detailed balanced Markov processes can be composed, thereby building up more complex processes from these open building blocks. The fact that `black-boxing'  is accomplished via a functor means that the behavior of a composite Markov process can be built up from the composite behaviors of the open Markov processes from which it is built. In this paper we illustrate how this framework can be utilized to study linear master equation systems far from equilibrium with a particular emphasis on the modeling of biological phenomena. 

Markovian or master equation systems have a long history of being used to model and understand biological systems. We make no attempt to provide a complete review of this line of work. Schnakenberg, in his paper on networked master equation systems, defines the entropy production in a Markov process and shows that a quantity related to entropy serves as a Lyapunov function for master equation systems \cite{SchnakenRev}.  His book \cite{SchnakenBook} provides a number of biochemical applications of networked master equation systems. Oster, Perelson and Katchalsky developed a theory of `networked thermodynamics' \cite{OPK}, which they went on to apply to the study of biological systems \cite{OPKBio}. Following the untimely passing of Katchalsky, Perelson and Oster went on to extend this work into the realm of chemical reactions \cite{OPChem}.

Starting in the 1970's, T.\ L.\ Hill spearheaded a line of research focused on what he called `free energy transduction' in biology. A shortened and updated form of his 1977 text on the subject \cite{Hill} was republished in 2005 \cite{HillDover}. Hill applied various techniques, such as the use of the cycle basis, in the analysis of biological systems. His model of muscle contraction provides one example \cite{HillScience}.

One quantity central to the study of non-equilibrium systems is the rate of entropy production \cite{DeGrootM, Prigogine, Lindblad, GP}. Prigogine's principle of minimum entropy production \cite{PrigogineEnt} asserts that for non-equilibrium steady states that are near equilibrium, entropy production is minimized. This is an approximate principle that is obtained by linearizing the relevant equations about an equilibrium state. In fact, for open Markov processes, non-equilibrium steady states are governed by a \emph{different} minimum principle that holds \emph{exactly}, arbitrarily far from equilibrium. We show that for fixed boundary conditions, non-equilibrium steady states minimize a quantity we call `dissipation'. If the populations of the non-equilibrium steady state are close to the population of the underlying detailed balanced equilibrium, one can show that dissipation is close to the rate of change of relative entropy plus a boundary term. Dissipation is in fact related to the Glansdorff-Prigogine criterion which states that a non-equilibrium steady state is stable if the second order variation of the entropy production is non-negative \cite{GP, SchnakenRev}.  

Starting in the 1990's, the Qians and their collaborators developed a school studying non-equilibrium steady states, publishing a number of articles and books on the topic \cite{Qians}. More recently, results concerning fluctuations have been extended to master equation systems \cite{AndrieuxGaspard}. In the past two decades, Hong Qian of the University of Washington and collaborators have published numerous results on non-equilibrium thermodynamics, biology and related topics \cite{Qian1, Qian2, Qian3}.

This paper is part of a larger project which uses category theory to unify a variety of diagrammatic approaches found across the sciences including, but not limited to, electrical circuits, control theory and bond graphs \cite{BaezFongCirc, BaezEberleControl}. We hope that the categorical approach will shed new light on each of these subjects as well as their interrelation, particularly as we generalize the results presented in this and recent papers to the more general, non-linear, setting of open chemical reaction networks. 

\section{The category of open detailed balanced Markov processes}\label{sec:DetBalMark}

In this section we describe how open detailed balanced Markov processes are the morphisms in a certain type of symmetric, monoidal, dagger-compact category. In previous work, Baez, Fong and the author \cite{BaezFongP} used the framework of decorated cospans \cite{Fong} to construct the category $\DetBalMark$. Here we give an intuitive description of this category and refer to those papers for the mathematical details. 

An object in $\DetBalMark$ is a \define{finite set with populations}, i.e. a finite set $X$ together with a map $p_X \maps X \to [0,\infty)$ assigning a population $p_i \in [0,\infty)$ to each element $i \in X$. A morphism $M \maps (X,p_X) \to (Y,p_Y)$ consists of an open detailed balanced Markov process together with \define{input} and \define{output} maps $i \maps (X,p_X) \to (V,q)$ and $o \maps (Y,p_Y) \to (V,q)$ which \define{preserve population} so that $p_X = iq$ and $p_Y = oq$. The \define{boundary} $B\subseteq V$ of an open Markov process is the union of the images of the input and output maps $B=i(X) \cup o(Y) $.

One can draw an open detailed balanced Markov process as a labeled directed graph whose vertices are labeled by their equilibrium populations and with specified subsets of the vertices as the input and the output states. Recall our simple model of membrane diffusion as an open detailed balanced Markov process, which we now think of as a morphism from the input $X = \{A\}$ to the output $Y=\{C\}$: 
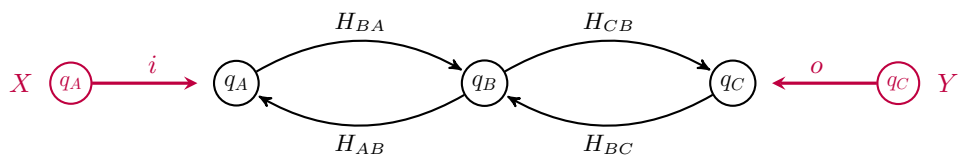
\[
\begin{tikzpicture}[->,>=stealth',shorten >=1pt,thick,scale=1.1]
\node[main node, scale=.65](A) at (-2,0) {$q_A$};
\node[main node, scale=.65](B) at (1,0) {$q_B$};
\node[main node, scale=.65](C) at (4,0) {$q_C$};
\node(input)[color=purple] at (-4.6,0) {$X$};
\node(input)[color=purple] at (6.6,0) {$Y$};
\node[terminal,scale=.6](X) at (-4,0) {$q_A$};
\node[terminal, scale=.6](Y) at (6,0) {$q_C$};

  \path[every node/.style={font=\sffamily\small}, shorten >=1pt]
    (A) edge [bend left] node[above] {$H_{BA}$} (B)
    (B) edge [bend left] node[below] {$H_{AB}$} (A) 
    (B) edge [bend left] node[above] {$H_{CB}$} (C)
    (C) edge [bend left] node[below] {$H_{BC}$} (B);

\path[color=purple, very thick, shorten >=6pt, ->, >=stealth] (X) edge  node[above] {$i$} (A);
\path[color=purple, very thick, shorten >=6pt, ->, >=stealth] (Y) edge  node[above] {$o$} (C);

\end{tikzpicture}
\]
This is a morphism in $\DetBalMark$ from $X$ to $Y$ where $X$ and $Y$ are finite sets with populations. In this simple example, $X$ and $Y$ both contain a single element, namely $A$ and $C$ respectively. Suppose we had another such membrane as depicted in Figure \ref{fig:membrain2}.
\begin{figure}[h!]
\begin{center}
\begin{tikzpicture}
	\begin{pgfonlayer}{nodelayer}
		\node [style=empty] (0) at (-0.5, -0) {};
		\node [style=shadecircle] (1) at (-2, -1) {};
		\node [style=shadecircle] (2) at (2.5, -1) {};
		\node [style=empty] (3) at (0.5, -0) {};
		\node [style=shadecircle] (4) at (-2, 1) {};
		\node [style=shadecircle] (5) at (-0.5, -1) {};
		\node [style=empty] (6) at (-3, 1) {};
		\node [style=empty] (7) at (-3.75, 1) {};
		\node [style=empty] (8) at (1, -0) {};
		\node [style=shadecircle] (10) at (-1.5, -1) {};
		\node [style=shadecircle] (11) at (0.5, -1) {};
		\node [style=shadecircle] (12) at (1.5, 1) {};
		\node [style=empty] (13) at (0, -0) {};
		\node [style=empty] (15) at (-3.75, -1) {};
		\node [style=shadecircle] (16) at (0.5, 1) {};
		\node [style=shadecircle] (17) at (-2.5, 1) {};
		\node [style=empty] (18) at (2, -0) {};
		\node [style=empty] (19) at (-5.25, 1) {};
		\node [style=shadecircle] (20) at (0, 1) {};
		\node [style=empty] (21) at (-4.5, 1) {};
		\node [style=empty] (22) at (-5.25, -1) {};
		\node [style=shadecircle] (23) at (-1.5, 1) {};
		\node [style=shadecircle] (24) at (1.5, -1) {};
		\node [style=shadecircle] (25) at (2, 1) {};
		\node [style=shadecircle] (26) at (1, 1) {};
		\node [style=shadecircle] (27) at (1, -1) {};
		\node [style=empty] (28) at (-3, -1) {};
		\node [style=empty] (29) at (1.5, -0) {};
		\node [style=shadecircle] (30) at (2.5, 1) {};
		\node [style=empty] (31) at (-2.5, -0) {};
		\node [style=empty] (32) at (-4.5, -1) {};
		\node [style=empty] (33) at (-1, -0) {};
		\node [style=shadecircle] (34) at (-1, 1) {};
		\node [style=shadecircle] (35) at (-0.5, 1) {};
		\node [style=empty] (36) at (-6, 1) {};
		\node [style=empty] (37) at (2.5, -0) {};
		\node [style=shadecircle] (38) at (0, -1) {};
		\node [style=empty] (39) at (-1.5, -0) {};
		\node [style=empty] (40) at (-6, -1) {};
		\node [style=shadecircle] (41) at (2, -1) {};
		\node [style=shadecircle] (42) at (-1, -1) {};
		\node [style=shadecircle] (44) at (-2.5, -1) {};
		\node [style=empty] (45) at (-2, -0) {};
		\node [style=empty] (9) at (-4.5, 1.5) {$C'$};
		\node [style=empty] (14) at (-4.5, -0) {$D$};
		\node [style=empty] (43) at (-4.5, -1.5) {$E$};
	\end{pgfonlayer}
	\begin{pgfonlayer}{edgelayer}
		\draw [style=lipid, bend right=15, looseness=1.00] (17) to (31);
		\draw [style=lipid, bend left=15, looseness=1.00] (17) to (31);
		\draw [style=lipid, bend left=15, looseness=1.00] (44) to (31);
		\draw [style=lipid, bend right=15, looseness=1.00] (44) to (31);
		\draw [style=lipid, bend right=15, looseness=1.00] (4) to (45);
		\draw [style=lipid, bend left=15, looseness=1.00] (4) to (45);
		\draw [style=lipid, bend left=15, looseness=1.00] (1) to (45);
		\draw [style=lipid, bend right=15, looseness=1.00] (1) to (45);
		\draw [style=lipid, bend right=15, looseness=1.00] (23) to (39);
		\draw [style=lipid, bend left=15, looseness=1.00] (23) to (39);
		\draw [style=lipid, bend left=15, looseness=1.00] (10) to (39);
		\draw [style=lipid, bend right=15, looseness=1.00] (10) to (39);
		\draw [style=lipid, bend right=15, looseness=1.00] (34) to (33);
		\draw [style=lipid, bend left=15, looseness=1.00] (34) to (33);
		\draw [style=lipid, bend left=15, looseness=1.00] (42) to (33);
		\draw [style=lipid, bend right=15, looseness=1.00] (42) to (33);
		\draw [style=lipid, bend right=15, looseness=1.00] (35) to (0);
		\draw [style=lipid, bend left=15, looseness=1.00] (35) to (0);
		\draw [style=lipid, bend left=15, looseness=1.00] (5) to (0);
		\draw [style=lipid, bend right=15, looseness=1.00] (5) to (0);
		\draw [style=lipid, bend right=15, looseness=1.00] (20) to (13);
		\draw [style=lipid, bend left=15, looseness=1.00] (20) to (13);
		\draw [style=lipid, bend left=15, looseness=1.00] (38) to (13);
		\draw [style=lipid, bend right=15, looseness=1.00] (38) to (13);
		\draw [style=lipid, bend right=15, looseness=1.00] (16) to (3);
		\draw [style=lipid, bend left=15, looseness=1.00] (16) to (3);
		\draw [style=lipid, bend left=15, looseness=1.00] (11) to (3);
		\draw [style=lipid, bend right=15, looseness=1.00] (11) to (3);
		\draw [style=lipid, bend right=15, looseness=1.00] (26) to (8);
		\draw [style=lipid, bend left=15, looseness=1.00] (26) to (8);
		\draw [style=lipid, bend left=15, looseness=1.00] (27) to (8);
		\draw [style=lipid, bend right=15, looseness=1.00] (27) to (8);
		\draw [style=lipid, bend right=15, looseness=1.00] (12) to (29);
		\draw [style=lipid, bend left=15, looseness=1.00] (12) to (29);
		\draw [style=lipid, bend left=15, looseness=1.00] (24) to (29);
		\draw [style=lipid, bend right=15, looseness=1.00] (24) to (29);
		\draw [style=lipid, bend right=15, looseness=1.00] (25) to (18);
		\draw [style=lipid, bend left=15, looseness=1.00] (25) to (18);
		\draw [style=lipid, bend left=15, looseness=1.00] (41) to (18);
		\draw [style=lipid, bend right=15, looseness=1.00] (41) to (18);
		\draw [style=lipid, bend right=15, looseness=1.00] (30) to (37);
		\draw [style=lipid, bend left=15, looseness=1.00] (30) to (37);
		\draw [style=lipid, bend left=15, looseness=1.00] (2) to (37);
		\draw [style=lipid, bend right=15, looseness=1.00] (2) to (37);
		\draw [style=lipid] (6) to (7);
		\draw [style=lipid] (7) to (21);
		\draw [style=lipid] (21) to (19);
		\draw [style=lipid] (19) to (36);
		\draw [style=lipid] (28) to (15);
		\draw [style=lipid] (15) to (32);
		\draw [style=lipid] (32) to (22);
		\draw [style=lipid] (22) to (40);
			\end{pgfonlayer}
\end{tikzpicture}
\caption{Another layer of membrane whose interior population is labeled by $D$ and whose exterior populations are labeled by $C'$ and $E$.}
\label{fig:membrain2}
\end{center}
\end{figure}
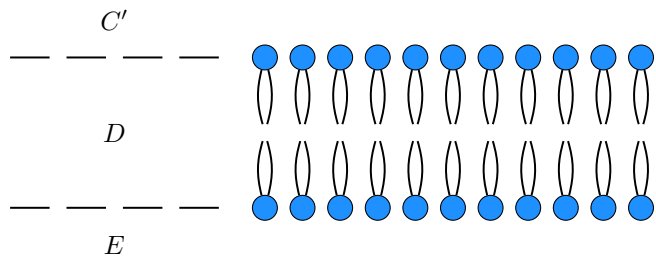
This is a morphism in $\DetBalMark$ from with input $Y=\{C'\}$ and output $Z=\{E\}$. Two open detailed balanced Markov processes can be composed if the detailed balanced equilibrium populations at the outputs of one match the detailed balanced equilibrium populations at the inputs of the other. This requirement guarantees that the composite of two open detailed balanced Markov process still admits a detailed balanced equilibrium.
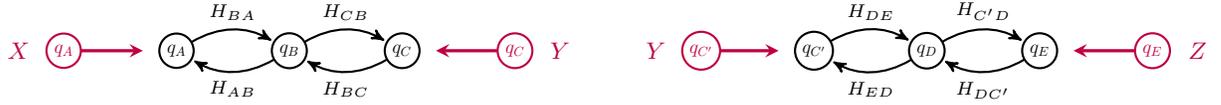
\[
\begin{tikzpicture}[->,>=stealth',shorten >=1pt,thick,scale=.75, font=\small]
\node[main node, scale=.5](A) at (-2,0) {$q_A$};
\node[main node, scale=.5](B) at (0,0) {$q_B$};
\node[main node, scale=.5](C) at (2,0) {$q_C$};
\node(input)[color=purple] at (-4.8,0) {$X$};
\node(input)[color=purple] at (4.8,0) {$Y$};
\node[terminal,scale=.5](X) at (-4,0) {$q_A$};
\node[terminal, scale=.5](Y) at (4,0) {$q_C$};

  \path[every node/.style={font=\scriptsize}, shorten >=1pt]
    (A) edge [bend left] node[above] {$H_{BA}$} (B)
    (B) edge [bend left] node[below] {$H_{AB}$} (A) 
    (B) edge [bend left] node[above] {$H_{CB}$} (C)
    (C) edge [bend left] node[below] {$H_{BC}$} (B);

\path[color=purple, very thick, shorten >=6pt, ->, >=stealth] (X) edge (A);
\path[color=purple, very thick, shorten >=6pt, ->, >=stealth] (Y) edge (C);

\node[main node, scale=.5](C') at (9.3,0) {$q_{C'}$};
\node[main node, scale=.5](B') at (11.3,0) {$q_D$};
\node[main node, scale=.5](A') at (13.3,0) {$q_{E}$};
\node(input)[color=purple] at (6.5,0) {$Y$};
\node(input)[color=purple] at (16.1,0) {$Z$};
\node[terminal,scale=.5](X') at (7.3,0) {$q_{C'}$};
\node[terminal, scale=.5](Y') at (15.3,0) {$q_{E}$};

  \path[every node/.style={font=\scriptsize}, shorten >=1pt]
    (C') edge [bend left] node[above] {$H_{DE}$} (B')
    (B') edge [bend left] node[below] {$H_{ED}$} (C') 
    (B') edge [bend left] node[above] {$H_{C'D}$} (A')
    (A') edge [bend left] node[below] {$H_{DC'}$} (B');

\path[color=purple, very thick, shorten >=6pt, ->, >=stealth] (X') edge (C');
\path[color=purple, very thick, shorten >=6pt, ->, >=stealth] (Y') edge (A');

\end{tikzpicture}
\]

If $q_C=q_{C'}$ in our two membrane models we can compose them by identifying $C$ with $C'$ to yield an open detailed balanced Markov process modeling the diffusion of neutral particles across membranes arranged in series: 
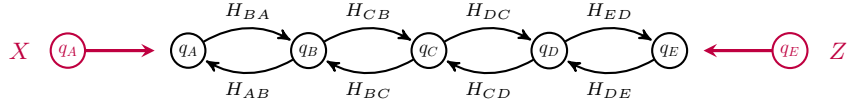
\[
\begin{tikzpicture}[->,>=stealth',shorten >=1pt,thick,scale=.8, font=\small]
\node[main node, scale=.5](A) at (-2,0) {$q_A$};
\node[main node, scale=.5](B) at (0,0) {$q_B$};
\node[main node, scale=.5](C) at (2,0) {$q_C$};
\node(input)[color=purple] at (-4.8,0) {$X$};
\node[terminal,scale=.5](X) at (-4,0) {$q_A$};

  \path[every node/.style={font=\scriptsize}, shorten >=1pt]
    (A) edge [bend left] node[above] {$H_{BA}$} (B)
    (B) edge [bend left] node[below] {$H_{AB}$} (A) 
    (B) edge [bend left] node[above] {$H_{CB}$} (C)
    (C) edge [bend left] node[below] {$H_{BC}$} (B);

\path[color=purple, very thick, shorten >=6pt, ->, >=stealth] (X) edge (A);

\node[main node, scale=.5](D) at (4,0) {$q_D$};
\node[main node, scale=.5](E) at (6,0) {$q_{E}$};
\node(input)[color=purple] at (8.8,0) {$Z$};
\node[terminal, scale=.5](Y') at (8,0) {$q_{E}$};

  \path[every node/.style={font=\scriptsize}, shorten >=1pt]
    (C) edge [bend left] node[above] {$H_{DC}$} (D)
    (D) edge [bend left] node[below] {$H_{CD}$} (C) 
    (D) edge [bend left] node[above] {$H_{ED}$} (E)
    (E) edge [bend left] node[below] {$H_{DE}$} (D);

\path[color=purple, very thick, shorten >=6pt, ->, >=stealth] (Y') edge (E);

\end{tikzpicture}
\]

Notice that the states corresponding to $C$ and $C'$ in each process have been identified and become internal states in the composite which is a morphism from $X=\{A\}$ to $Z=\{E\}$. This open Markov process can be thought of as modeling the diffusion across two membranes in series, see Figure \ref{fig:membrain3}.

\begin{figure}[h]
\begin{center}
\begin{tikzpicture}
	\begin{pgfonlayer}{nodelayer}
		\node [style=empty] (0) at (-0.5, -0) {};
		\node [style=shadecircle] (1) at (-2, -1) {};
		\node [style=shadecircle] (2) at (2.5, -1) {};
		\node [style=empty] (3) at (0.5, -0) {};
		\node [style=shadecircle] (4) at (-2, 1) {};
		\node [style=shadecircle] (5) at (-0.5, -1) {};
		\node [style=empty] (6) at (-3, 1) {};
		\node [style=empty] (7) at (-3.75, 1) {};
		\node [style=empty] (8) at (1, -0) {};
		\node [style=shadecircle] (10) at (-1.5, -1) {};
		\node [style=shadecircle] (11) at (0.5, -1) {};
		\node [style=shadecircle] (12) at (1.5, 1) {};
		\node [style=empty] (13) at (0, -0) {};
		\node [style=empty] (15) at (-3.75, -1) {};
		\node [style=shadecircle] (16) at (0.5, 1) {};
		\node [style=shadecircle] (17) at (-2.5, 1) {};
		\node [style=empty] (18) at (2, -0) {};
		\node [style=empty] (19) at (-5.25, 1) {};
		\node [style=shadecircle] (20) at (0, 1) {};
		\node [style=empty] (21) at (-4.5, 1) {};
		\node [style=empty] (22) at (-5.25, -1) {};
		\node [style=shadecircle] (23) at (-1.5, 1) {};
		\node [style=shadecircle] (24) at (1.5, -1) {};
		\node [style=shadecircle] (25) at (2, 1) {};
		\node [style=shadecircle] (26) at (1, 1) {};
		\node [style=shadecircle] (27) at (1, -1) {};
		\node [style=empty] (28) at (-3, -1) {};
		\node [style=empty] (29) at (1.5, -0) {};
		\node [style=shadecircle] (30) at (2.5, 1) {};
		\node [style=empty] (31) at (-2.5, -0) {};
		\node [style=empty] (32) at (-4.5, -1) {};
		\node [style=empty] (33) at (-1, -0) {};
		\node [style=shadecircle] (34) at (-1, 1) {};
		\node [style=shadecircle] (35) at (-0.5, 1) {};
		\node [style=empty] (36) at (-6, 1) {};
		\node [style=empty] (37) at (2.5, -0) {};
		\node [style=shadecircle] (38) at (0, -1) {};
		\node [style=empty] (39) at (-1.5, -0) {};
		\node [style=empty] (40) at (-6, -1) {};
		\node [style=shadecircle] (41) at (2, -1) {};
		\node [style=shadecircle] (42) at (-1, -1) {};
		\node [style=shadecircle] (44) at (-2.5, -1) {};
		\node [style=empty] (45) at (-2, -0) {};
		\node [style=empty] (9) at (-4.5, 1.5) {$A$};
		\node [style=empty] (14) at (-4.5, -0) {$B$};
		\node [style=empty] (43) at (-4.5, -2) {$C$};
	\end{pgfonlayer}
	\begin{pgfonlayer}{edgelayer}
		\draw [style=lipid, bend right=15, looseness=1.00] (17) to (31);
		\draw [style=lipid, bend left=15, looseness=1.00] (17) to (31);
		\draw [style=lipid, bend left=15, looseness=1.00] (44) to (31);
		\draw [style=lipid, bend right=15, looseness=1.00] (44) to (31);
		\draw [style=lipid, bend right=15, looseness=1.00] (4) to (45);
		\draw [style=lipid, bend left=15, looseness=1.00] (4) to (45);
		\draw [style=lipid, bend left=15, looseness=1.00] (1) to (45);
		\draw [style=lipid, bend right=15, looseness=1.00] (1) to (45);
		\draw [style=lipid, bend right=15, looseness=1.00] (23) to (39);
		\draw [style=lipid, bend left=15, looseness=1.00] (23) to (39);
		\draw [style=lipid, bend left=15, looseness=1.00] (10) to (39);
		\draw [style=lipid, bend right=15, looseness=1.00] (10) to (39);
		\draw [style=lipid, bend right=15, looseness=1.00] (34) to (33);
		\draw [style=lipid, bend left=15, looseness=1.00] (34) to (33);
		\draw [style=lipid, bend left=15, looseness=1.00] (42) to (33);
		\draw [style=lipid, bend right=15, looseness=1.00] (42) to (33);
		\draw [style=lipid, bend right=15, looseness=1.00] (35) to (0);
		\draw [style=lipid, bend left=15, looseness=1.00] (35) to (0);
		\draw [style=lipid, bend left=15, looseness=1.00] (5) to (0);
		\draw [style=lipid, bend right=15, looseness=1.00] (5) to (0);
		\draw [style=lipid, bend right=15, looseness=1.00] (20) to (13);
		\draw [style=lipid, bend left=15, looseness=1.00] (20) to (13);
		\draw [style=lipid, bend left=15, looseness=1.00] (38) to (13);
		\draw [style=lipid, bend right=15, looseness=1.00] (38) to (13);
		\draw [style=lipid, bend right=15, looseness=1.00] (16) to (3);
		\draw [style=lipid, bend left=15, looseness=1.00] (16) to (3);
		\draw [style=lipid, bend left=15, looseness=1.00] (11) to (3);
		\draw [style=lipid, bend right=15, looseness=1.00] (11) to (3);
		\draw [style=lipid, bend right=15, looseness=1.00] (26) to (8);
		\draw [style=lipid, bend left=15, looseness=1.00] (26) to (8);
		\draw [style=lipid, bend left=15, looseness=1.00] (27) to (8);
		\draw [style=lipid, bend right=15, looseness=1.00] (27) to (8);
		\draw [style=lipid, bend right=15, looseness=1.00] (12) to (29);
		\draw [style=lipid, bend left=15, looseness=1.00] (12) to (29);
		\draw [style=lipid, bend left=15, looseness=1.00] (24) to (29);
		\draw [style=lipid, bend right=15, looseness=1.00] (24) to (29);
		\draw [style=lipid, bend right=15, looseness=1.00] (25) to (18);
		\draw [style=lipid, bend left=15, looseness=1.00] (25) to (18);
		\draw [style=lipid, bend left=15, looseness=1.00] (41) to (18);
		\draw [style=lipid, bend right=15, looseness=1.00] (41) to (18);
		\draw [style=lipid, bend right=15, looseness=1.00] (30) to (37);
		\draw [style=lipid, bend left=15, looseness=1.00] (30) to (37);
		\draw [style=lipid, bend left=15, looseness=1.00] (2) to (37);
		\draw [style=lipid, bend right=15, looseness=1.00] (2) to (37);
		\draw [style=lipid] (6) to (7);
		\draw [style=lipid] (7) to (21);
		\draw [style=lipid] (21) to (19);
		\draw [style=lipid] (19) to (36);
		\draw [style=lipid] (28) to (15);
		\draw [style=lipid] (15) to (32);
		\draw [style=lipid] (32) to (22);
		\draw [style=lipid] (22) to (40);
			\end{pgfonlayer}
\end{tikzpicture}
\vskip 1em
\begin{tikzpicture}
	\begin{pgfonlayer}{nodelayer}
		\node [style=empty] (0) at (-0.5, -0) {};
		\node [style=shadecircle] (1) at (-2, -1) {};
		\node [style=shadecircle] (2) at (2.5, -1) {};
		\node [style=empty] (3) at (0.5, -0) {};
		\node [style=shadecircle] (4) at (-2, 1) {};
		\node [style=shadecircle] (5) at (-0.5, -1) {};
		\node [style=empty] (6) at (-3, 1) {};
		\node [style=empty] (7) at (-3.75, 1) {};
		\node [style=empty] (8) at (1, -0) {};
		\node [style=shadecircle] (10) at (-1.5, -1) {};
		\node [style=shadecircle] (11) at (0.5, -1) {};
		\node [style=shadecircle] (12) at (1.5, 1) {};
		\node [style=empty] (13) at (0, -0) {};
		\node [style=empty] (15) at (-3.75, -1) {};
		\node [style=shadecircle] (16) at (0.5, 1) {};
		\node [style=shadecircle] (17) at (-2.5, 1) {};
		\node [style=empty] (18) at (2, -0) {};
		\node [style=empty] (19) at (-5.25, 1) {};
		\node [style=shadecircle] (20) at (0, 1) {};
		\node [style=empty] (21) at (-4.5, 1) {};
		\node [style=empty] (22) at (-5.25, -1) {};
		\node [style=shadecircle] (23) at (-1.5, 1) {};
		\node [style=shadecircle] (24) at (1.5, -1) {};
		\node [style=shadecircle] (25) at (2, 1) {};
		\node [style=shadecircle] (26) at (1, 1) {};
		\node [style=shadecircle] (27) at (1, -1) {};
		\node [style=empty] (28) at (-3, -1) {};
		\node [style=empty] (29) at (1.5, -0) {};
		\node [style=shadecircle] (30) at (2.5, 1) {};
		\node [style=empty] (31) at (-2.5, -0) {};
		\node [style=empty] (32) at (-4.5, -1) {};
		\node [style=empty] (33) at (-1, -0) {};
		\node [style=shadecircle] (34) at (-1, 1) {};
		\node [style=shadecircle] (35) at (-0.5, 1) {};
		\node [style=empty] (36) at (-6, 1) {};
		\node [style=empty] (37) at (2.5, -0) {};
		\node [style=shadecircle] (38) at (0, -1) {};
		\node [style=empty] (39) at (-1.5, -0) {};
		\node [style=empty] (40) at (-6, -1) {};
		\node [style=shadecircle] (41) at (2, -1) {};
		\node [style=shadecircle] (42) at (-1, -1) {};
		\node [style=shadecircle] (44) at (-2.5, -1) {};
		\node [style=empty] (45) at (-2, -0) {};

		\node [style=empty] (14) at (-4.5, -0) {$D$};
		\node [style=empty] (43) at (-4.5, -1.5) {$E$};
	\end{pgfonlayer}
	\begin{pgfonlayer}{edgelayer}
		\draw [style=lipid, bend right=15, looseness=1.00] (17) to (31);
		\draw [style=lipid, bend left=15, looseness=1.00] (17) to (31);
		\draw [style=lipid, bend left=15, looseness=1.00] (44) to (31);
		\draw [style=lipid, bend right=15, looseness=1.00] (44) to (31);
		\draw [style=lipid, bend right=15, looseness=1.00] (4) to (45);
		\draw [style=lipid, bend left=15, looseness=1.00] (4) to (45);
		\draw [style=lipid, bend left=15, looseness=1.00] (1) to (45);
		\draw [style=lipid, bend right=15, looseness=1.00] (1) to (45);
		\draw [style=lipid, bend right=15, looseness=1.00] (23) to (39);
		\draw [style=lipid, bend left=15, looseness=1.00] (23) to (39);
		\draw [style=lipid, bend left=15, looseness=1.00] (10) to (39);
		\draw [style=lipid, bend right=15, looseness=1.00] (10) to (39);
		\draw [style=lipid, bend right=15, looseness=1.00] (34) to (33);
		\draw [style=lipid, bend left=15, looseness=1.00] (34) to (33);
		\draw [style=lipid, bend left=15, looseness=1.00] (42) to (33);
		\draw [style=lipid, bend right=15, looseness=1.00] (42) to (33);
		\draw [style=lipid, bend right=15, looseness=1.00] (35) to (0);
		\draw [style=lipid, bend left=15, looseness=1.00] (35) to (0);
		\draw [style=lipid, bend left=15, looseness=1.00] (5) to (0);
		\draw [style=lipid, bend right=15, looseness=1.00] (5) to (0);
		\draw [style=lipid, bend right=15, looseness=1.00] (20) to (13);
		\draw [style=lipid, bend left=15, looseness=1.00] (20) to (13);
		\draw [style=lipid, bend left=15, looseness=1.00] (38) to (13);
		\draw [style=lipid, bend right=15, looseness=1.00] (38) to (13);
		\draw [style=lipid, bend right=15, looseness=1.00] (16) to (3);
		\draw [style=lipid, bend left=15, looseness=1.00] (16) to (3);
		\draw [style=lipid, bend left=15, looseness=1.00] (11) to (3);
		\draw [style=lipid, bend right=15, looseness=1.00] (11) to (3);
		\draw [style=lipid, bend right=15, looseness=1.00] (26) to (8);
		\draw [style=lipid, bend left=15, looseness=1.00] (26) to (8);
		\draw [style=lipid, bend left=15, looseness=1.00] (27) to (8);
		\draw [style=lipid, bend right=15, looseness=1.00] (27) to (8);
		\draw [style=lipid, bend right=15, looseness=1.00] (12) to (29);
		\draw [style=lipid, bend left=15, looseness=1.00] (12) to (29);
		\draw [style=lipid, bend left=15, looseness=1.00] (24) to (29);
		\draw [style=lipid, bend right=15, looseness=1.00] (24) to (29);
		\draw [style=lipid, bend right=15, looseness=1.00] (25) to (18);
		\draw [style=lipid, bend left=15, looseness=1.00] (25) to (18);
		\draw [style=lipid, bend left=15, looseness=1.00] (41) to (18);
		\draw [style=lipid, bend right=15, looseness=1.00] (41) to (18);
		\draw [style=lipid, bend right=15, looseness=1.00] (30) to (37);
		\draw [style=lipid, bend left=15, looseness=1.00] (30) to (37);
		\draw [style=lipid, bend left=15, looseness=1.00] (2) to (37);
		\draw [style=lipid, bend right=15, looseness=1.00] (2) to (37);
		\draw [style=lipid] (6) to (7);
		\draw [style=lipid] (7) to (21);
		\draw [style=lipid] (21) to (19);
		\draw [style=lipid] (19) to (36);
		\draw [style=lipid] (28) to (15);
		\draw [style=lipid] (15) to (32);
		\draw [style=lipid] (32) to (22);
		\draw [style=lipid] (22) to (40);
			\end{pgfonlayer}
\end{tikzpicture}
\caption{A depiction of two membranes arranged in series.} 
\label{fig:membrain3}
\end{center}
\end{figure}
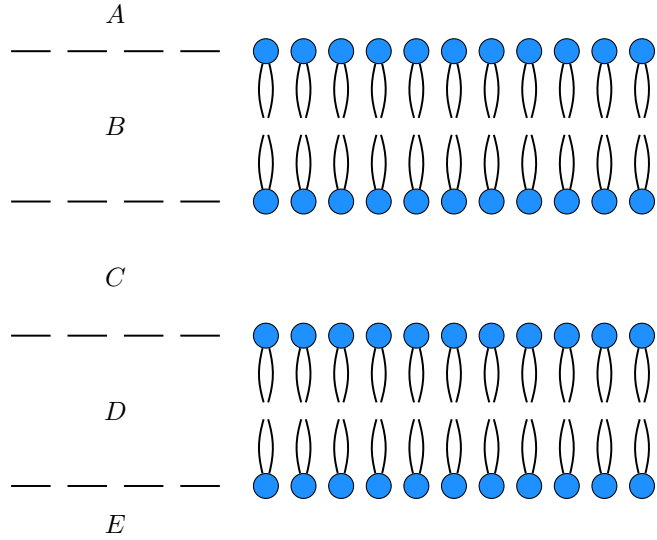

One can `black-box' an open detailed balanced Markov process by converting it into an electrical circuit, applying the already known black-boxing functor for electrical circuits \cite{BaezFongCirc} and translating the result back into the language of open Markov processes \cite{BaezFongP}. The key step in this process is the construction of a quadratic form which we call `dissipation', analogous to power in electrical circuits, which is minimized when the populations of an open Markov process are in a steady state. 

\section{Principle of minimum dissipation}\label{sec:dissipation}

Here we show that by externally fixing the populations at boundary states, one induces steady states which minimize a quadratic form which we call `dissipation.' 
\begin{defn} Given an open detailed balanced Markov process we define the \define{dissipation functional} of a population distribution $p$ to be
\[ D(p) = \frac{1}{2}\sum_{i,j} H_{ij}q_j \left( \frac{p_j}{q_j} - \frac{p_i}{q_i} \right)^2. \]
\end{defn}
Given boundary populations $b \in R^B$, we can minimize this functional over all $p$ which agree on the boundary. Differentiating the dissipation functional with respect to an internal population, we get
\[ \frac{\partial D(p)}{\partial p_n} = -2\sum_j H_{nj}\frac{p_j}{q_n}. \]
Multiplying by $\frac{q_n}{2}$ yields
\[ \frac{q_n}{2} \frac{ \partial D(p)}{\partial p_n} = -\sum_j H_{nj} p_j, \]
where we recognize the right-hand side from the open master equation for internal states. We see that for fixed boundary populations, the conditions for $p$ to be a steady state, namely that
\[ \frac{dp_i}{dt} = 0 \ \  \text{for all} \ \ i \in V,\] is equivalent to the condition that 
\[ \frac{\partial D(p)}{\partial p_n} = 0 \ \  \text{for all} \ \ n \in V-B.\]

\begin{defn}
We say a population distribution obeys the \define{principle of minimum dissipation with boundary population $b$} if $p$ minimizes $D(p)$ subject to the constraint that $p|_b = b$.
\end{defn} 
With this we can state the following theorem:
\begin{thm}
A population distribution $p \in \R^V$ is a steady state with boundary population $b \in \R^B$ if and only if $p$ obeys the principle of minimum dissipation with boundary population $b$.
\end{thm}
\begin{proof}
This follows from Theorem 28 in \cite{BaezFongP}.
\end{proof}

Given specified boundary populations, one can compute the steady state boundary flows by minimizing the dissipation subject to the boundary conditions.
\begin{defn}
We call a population-flow pair a \define{steady state population-flow pair} if the flows arise from a population distribution which obeys the principle of minimum dissipation. 
\end{defn}
\begin{defn}
The \define{behavior} of an open detailed balanced Markov process with boundary $B$ is the set of all steady state population-flow pairs $(p_B, J_B)$ along the boundary.
\end{defn}
Indeed, there is a functor $\square \maps \DetBalMark \to \LinRel$ which maps open detailed balanced Markov processes to their steady state behaviors. This is the main result of our previous paper \cite{BaezFongP}. The fact that this is a functor means that the behavior of a composite open detailed balanced Markov process can be computed as the composite of the behaviors.

\section{Dissipation and Entropy Production} \label{sec:rentropy}

In the last section, we saw that non-equilibrium steady states with fixed boundary populations minimize the dissipation. In this section we relate the dissipation to a divergence between population distributions known in various circles as the relative entropy, relative information or the Kullback-Leibler divergence. The relative entropy is not symmetric and violates the triangle inequality, which is why it is called a `divergence' rather than a metric, or distance function.  We show that for population distributions near a detailed balanced equilibrium, the rate of change of the relative entropy is approximately equal to the dissipation plus a `boundary term'. 

The \define{relative entropy} of two distributions $p,q$ is given by
\[ I(p,q) = \sum_i p_i \ln \left( \frac{p_i}{q_i} \right). \]
It is well known that for a closed Markov process admitting a detailed balanced equilibrium, the relative entropy with respect to this detailed balanced equilibrium distribution is monotonically decreasing with time, see for instance \cite{Kelly}. There is an unfortunate sign convention in the definition of relative entropy: while entropy is typically increasing, relative entropy typically decreases.  More generally, the relative entropy between any two population distributions is non-increasing in a closed Markov process. 

In an open Markov process, the sign of the rate of change of relative entropy is indeterminate. Consider an open Markov process $(V,B,H)$. For any two population distributions $p(t)$ and $q(t)$ which obey the open master equation let us introduce the quantities
\[ \frac{Dp_i}{Dt} = \frac{d p_i}{dt} - \sum_{j \in V} H_{ij} p_j \]
and
\[ \frac{Dq_i}{Dt} = \frac{d q_i}{dt} - \sum_{j \in V} H_{ij} q_j, \]
which measure how much the time derivatives of $p(t)$ and $q(t)$ fail to obey the master equation. Notice that $\frac{Dp_i}{Dt} = 0 $ for $ i \in V-B$, as the populations of internal states evolve according to the master equation. In terms of these quantities, the rate of change of relative entropy for an open Markov process can be written as
\[ \frac{d}{dt} I(p(t),q(t)) = \sum_{i,j \in V} H_{ij}p_j \left( \ln \left(\frac{p_i}{q_i} \right) - \frac{p_i q_j}{q_i p_j} \right) + \sum_{i \in B} \frac{Dp_i}{Dt} \frac{ \partial I}{\partial p_i} + \frac{Dq_i}{Dt} \frac{\partial I}{\partial q_i}. \]
The first term is the rate of change of relative entropy for a closed Markov process. This is less than or equal to zero \cite{BaezP, Pollard}. Thus, the rate of change of relative entropy in an open Markov process satisfies 
\[ \frac{d}{dt} I(p(t), q(t)) \leq \sum_{i \in B} \frac{Dp_i}{Dt} \frac{ \partial I}{\partial p_i} + \frac{Dq_i}{Dt} \frac{\partial I}{\partial q_i}.\]
This inequality tells us that the rate of change of relative entropy in an open Markov processes is bounded by the rate at which relative entropy flows through its boundary. If $q$ is an equilibrium solution of the master equation 
\[ \frac{dq}{dt} = Hq = 0, \] then the rate of change of relative entropy can be written as
\[ \frac{d}{dt} I(p(t),q) = \sum_{i,j \in V} ( H_{ij}p_j - H_{ji}p_i) \ln\left( \frac{p_i q_j}{q_i p_j} \right
) + \sum_{i \in B} \frac{Dp_i}{Dt} \frac{ \partial I}{\partial p_i} \]
Furthermore, if $q$ satisfies detailed balance we can write this as
\[ \frac{d}{dt} I(p(t),q) = -\frac{1}{2} \sum_{i,j \in V} J_{ij} A_{ij}+ \sum_{i \in B} \frac{Dp_i}{Dt} \frac{ \partial I}{\partial p_i}., \]
where  
\[ J_{ij}(p) = H_{ij}p_j - H_{ji}p_i \]
is the \define{thermodynamic flux} from $j$ to $i$ and 
\[ A_{ij}(p) = \ln \left( \frac{H_{ij}p_j}{H_{ji}p_i} \right) \]
is the conjugate \define{thermodynamic force}. This quantity: 
\[ \frac{1}{2} \sum_{i,j \in V} J_{ij}A_{ij} \] is what Schnakenberg calls ``the rate of entropy production'' \cite{SchnakenRev}. This is always non-negative. Note that due to the sign convention in the definition of relative entropy, in the absence of the boundary term, a positive rate of entropy production corresponds to a decreasing relative entropy.

We shall shortly relate the rate of change of relative entropy to the dissipation for open detailed balanced Markov processes, but first let us consider the quantity $A_{ij}(p)$. It is the entropy production per unit flow from $j$ to $i$. If $J_{ij}(p)>0$, i.e. if there is a positive net flow of population from $j$ to $i$, then $A_{ij}(p)>0$. In addition, $J_{ij}(p)=0$ implies that $A_{ij}(p)=0$. Thus we see that this form of entropy production is, by definition, non-negative. 

In the realm of population dynamics, we can understand $A_{ij}(p)$ as the force resulting from a difference in chemical potential. Let us elaborate on this point to clarify the relation of our framework to the language of chemical potentials used in non-equilibrium thermodynamics. Suppose that we are dealing with only a single type of molecule or chemical species. The states could correspond to different locations of the molecule, as in our example of membrane transport. Another possibility is that each state correspond to a different internal configuration of the molecule. In this setting the chemical potential $\mu_i$ is related to the concentration of that chemical species in the following way:
\[ \mu_i = \mu_i^o + T \ln( c_i), \]
where $T$ is the temperature of the system in units where Boltzmann's constant is equal to one and $\mu_i^o$ is the standard chemical potential. The difference in chemical potential between two states gives the force associated with the flow of population which seeks to reduce this difference in chemical potential 
\[ \mu_j - \mu_i = \mu_j^o - \mu_i^o + T \ln\left( \frac{c_j}{c_i} \right). \] 
In general the concentration of the $i^{\text{th}}$ state is proportional to the population of that chemical species divided by the volume of the system $c_i = \frac{p_i}{V}$. In this case, the volumes cancel out in the ratio of concentrations and we have this relation between chemical potential differences and population differences:
\[ \mu_j - \mu_i = \mu_j^o - \mu_i^o + T \ln \left( \frac{p_j}{p_i} \right). \]
This potential difference vanishes when $p_i$ and $p_j$ are in equilibrium and we have
\[ 0 = \mu_j^o - \mu_i^o + T \ln \left( \frac{q_j}{q_i} \right), \]
or that
\[ \frac{q_j}{q_i} = e^{-\frac{\mu_j^o-\mu_i^o}{T}}. \]
If $q$ satisfies detailed balance, then this also gives an expression for the ratio of the transition rates $\frac{H_{ji}}{H_{ij}}$ in terms of the standard chemical potentials. Thus we can translate between differences in chemical potential and ratios of populations via the relation
\[ \mu_j - \mu_i = T \ln \left( \frac{p_j q_i}{q_j p_i} \right), \]
which if $q$ satisfies detailed balance gives
\[ \mu_j - \mu_i = T \ln \left( \frac{H_{ij}p_j}{H_{ji}p_i} \right). \]
We recognize the right hand side as the force $A_{ij}(p)$ times the temperature of the system $T$:
\[ \frac{\mu_j - \mu_i}{T} = A_{ij}(p) .\]

Let us return to our expression for $\frac{d}{dt} I(p(t),q)$ where $q$ is an equilibrium distribution:
\[ \frac{d}{dt} I(p(t),q) = -\frac{1}{2}\sum_{i,j \in V} \left( H_{ij}p_j - H_{ji}p_i \right) \ln \left( \frac{q_i p_j}{q_j p_i} \right) + \sum_{i \in B} \frac{Dp_i}{Dt}\frac{\partial I}{\partial p_i}. \]
Consider the situation in which $p$ is near to the equilibrium distribution $q$ in the sense that
\[ \frac{p_i}{q_i} = 1 + \epsilon_i \]
where $\epsilon_i \in \R$ is the deviation in the ratio $\frac{p_i}{q_i}$ from unity. We collect these deviations in a vector denoted by $\epsilon$. 
Expanding the logarithm to first order in $\epsilon$ we have that
\[ \frac{d}{dt} I(p(t),q) = -\frac{1}{2} \sum_{i,j \in V} \left(H_{ij} p_j - H_{ji}p_i \right) \left( \epsilon_j - \epsilon_i \right) + \sum_{i \in B} \frac{Dp_i}{Dt} \frac{\partial I}{\partial p_i} + O(\epsilon^2), \]
which gives
\[ \frac{d}{dt} I(p(t),q) = -\frac{1}{2} \sum_{i,j \in V} \left(H_{ij}p_j - H_{ji}p_i \right) \left(\frac{p_j}{q_j} - \frac{p_i}{q_i} \right) + \sum_{i \in B}  \frac{Dp_i}{Dt}\frac{\partial I}{\partial p_i} + O(\epsilon^2).\]
By $O(\epsilon^2)$ we mean a sum of terms of order $\epsilon_i^2$. 
When $q$ is a detailed balanced equilibrium we can rewrite this quantity as
\[ \frac{d}{dt} I(p(t),q) = -\frac{1}{2} \sum_{i,j} H_{ij}q_j \left( \frac{p_j}{q_j} - \frac{p_i}{q_i} \right)^2 + \sum_{i \in B} \frac{Dp_i}{Dt} \frac{\partial I}{\partial p_i} + O(\epsilon^2). \]
We recognize the first term as the negative of the dissipation $D(p)$ which yields
\[ \frac{d}{dt} I(p(t),q) = - D(p) + \sum_{i \in B} \frac{Dp_i}{Dt}\frac{\partial I}{\partial p_i} + O(\epsilon^2) . \]

We see that for open Markov processes, minimizing the dissipation approximately minimizes the rate of decrease of relative entropy plus a term which depends on the boundary populations. In the case that boundary populations are held fixed so that $\frac{dp_i}{dt} = 0 , \ i \in B$, we have that 
\[ \frac{Dp_i}{Dt} = -\sum_{j \in V} H_{ij} p_j, \ \ i \in B. \]
In this case, the rate of change of relative entropy can be written as
\[ \frac{d}{dt} I(p(t),q) = \sum_{i \in V-B} \frac{p_i}{q_i} \frac{dp_i}{dt} + 2 \sum_{i \in B} \frac{Dp_i}{Dt} + O(\epsilon^2). \]

Summarizing the results of this section, we have that for $p$ arbitrarily far from the detailed balanced equilibrium equilibrium $q$, the rate of relative entropy reduction can be written as 
\[ \frac{dI(p(t),q)}{dt}  = -\frac{1}{2}\sum_{i,j}  J_{ij}(p) A_{ij}(p)+ \sum_{i \in B} \frac{Dp_i}{Dt} \frac{\partial I}{\partial p_i}. \]
For $p$ in the vicinity of a detailed balanced equilibrium we have that
\[ \frac{dI(p(t),q)}{dt} = -D(p) + \sum_{i \in B} \frac{Dp_i}{Dt}\frac{\partial I}{\partial p_i} + O(\epsilon^2) \]
where $D(p)$ is the dissipation and $\epsilon_i = \frac{p_i}{q_i} - 1$ measures the deviations of the populations $p_i$ from their equilibrium values. We have seen that in a non-equilibrium steady state with fixed boundary populations, dissipation is minimized. We showed that for steady states near equilibirum, the rate of change of relative entropy is approximately equal to minus the dissipation plus a boundary term. Minimum dissipation coincides with minimum entropy production only in the limit $\epsilon \to 0$.

\section{Minimum Dissipation versus Minimum Entropy Production}

We return to our simple three-state example of membrane transport to illustrate the difference between populations which minimize dissipation and those which minimize entropy production:  
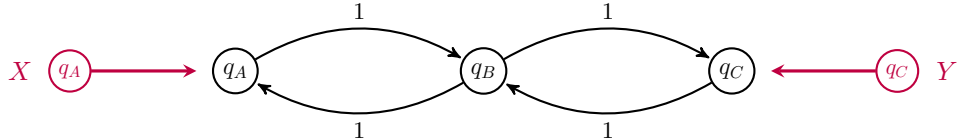
\[
\begin{tikzpicture}[->,>=stealth',shorten >=1pt,thick,scale=1.1]
\node[main node, scale=.65](A) at (-2,0) {$q_A$};
\node[main node, scale=.65](B) at (1,0) {$q_B$};
\node[main node, scale=.65](C) at (4,0) {$q_C$};
\node(input)[color=purple] at (-4.6,0) {$X$};
\node(input)[color=purple] at (6.6,0) {$Y$};
\node[terminal,scale=.6](X) at (-4,0) {$q_A$};
\node[terminal, scale=.6](Y) at (6,0) {$q_C$};

  \path[every node/.style={font=\sffamily\small}, shorten >=1pt]
    (A) edge [bend left] node[above] {$1$} (B)
    (B) edge [bend left] node[below] {$1$} (A) 
    (B) edge [bend left] node[above] {$1$} (C)
    (C) edge [bend left] node[below] {$1$} (B);

\path[color=purple, very thick, shorten >=6pt, ->, >=stealth] (X) edge (A);
\path[color=purple, very thick, shorten >=6pt, ->, >=stealth] (Y) edge (C);

\end{tikzpicture}
\]
For simplicity, we have set all transition rates equal to one. In this case, the detailed balance equilibrium distribution is uniform. We take $q_A = q_B = q_C = 1$. If the populations $p_A$ and $p_C$ are externally fixed, then the population $p_B$ which minimizes the dissipation is simply the arithmetic mean of the boundary populations
\[ p_B = \frac{p_A + p_C}{2}. \]

The rate of change of the relative entropy $I(p(t),q)$ where $q$ is the uniform detailed balanced equilibrium is given by
\begin{multline*} \frac{d}{dt} I(p(t),q) = \\
\underbrace{-(p_A - p_B) \ln  (\frac{p_A}{p_B} ) - (p_B - p_C) \ln ( \frac{p_B}{p_C} )}_{-\frac{1}{2}\sum_{i,j \in V} J_{ij}A_{ij}} 
+\underbrace{(p_A - p_B) ( \ln(p_A) + 1) + (p_C- p_B) ( \ln(p_C) +1 )}_{\sum_{i \in B} \frac{Dp_i}{Dt}\frac{\partial I}{\partial p_i}}. \end{multline*}
Differentiating this quantity with respect to $p_B$ for fixed $p_A$ and $p_C$ yields the condition
\[ \frac{p_A+p_C}{2p_B} - \ln(p_B) - 2 = 0. \]
The solution of this equation gives the population $p_B$ which extremizes the rate of change of relative entropy, namely
\[ p_B = \frac{p_A+p_C}{2W \left(\frac{(p_A+p_C)}{2} e^2 \right) }, \]
where $W(x)$ is the Lambert $W$-function or the omega function which satisfies the following relation
\[ x = W(x)e^{W(x)}. \]
The Lambert $W$-function is defined for $x \geq \frac{-1}{e}$ and double valued for $x \in [\frac{-1}{e},0)$. This simple example illustrates the difference between distributions which minimize dissipation subject to boundary constraints and those which extremize the rate of change of relative entropy. For fixed boundary populations, dissipation is minimized in steady states arbitrarily far from equilibrium. For steady states in the neighborhood of the detailed balanced equilibrium, the rate of change of relative entropy is approximately equal to minus the dissipation plus a boundary term.

\section{Discussion}

Treating Markov processes as morphisms in a category leads naturally to open systems which admit non-equilibrium steady states, even when the transition rates of the underlying process satisfy Kolmogorov's criterion. Microscopically, all reactions should be reversible with perhaps a large disparity between the forward and reverse rates. Nonetheless, it is clear that biological organisms are capable, at least locally, of storing free energy. This is typically accomplished via the interaction with other systems or the environment. In this paper, the environment served as a reservoir maintaining boundary populations at constant values. Since open Markov processes are morphisms in the category $\DetBalMark$, one can compose these open systems, thereby building up complicated systems in a systematic way. We saw that the non-equilibrium steady states which emerge minimize a quadratic form which depends on the deviation of the steady state populations from the populations of the underlying detailed balanced equilibrium. For steady states in the neighborhood of equilibrium, we saw that the dissipation is in fact the linear approximation of the rate of change of relative entropy with respect to a detailed balanced equilibrium plus a boundary term. In our framework, dissipation appears to be the fundamental quantity as it is minimized for non-equilibrium steady states arbitrarily far from equilibrium. There has been much work examining the regime of validity of Prigogine's principle of minimum entropy production \cite{Landauer, Landauer2, BMN}. In future work, we aim to generalize our framework for composing Markov processes to the non-linear regime of chemical reaction networks with an eye towards incorporating recent interesting results in the area \cite{PolettiniCRN}. We anticipate that the perspective achieved by viewing interacting systems as morphisms in a category will bring new insight to the study of living systems far from equilibrium.


\section*{Acknowledgments}
The author would like to thank John C. Baez for his help developing the ideas presented in this paper and improving the quality and clarity of their exposition.  The author also thanks Brendan Fong for many useful discussions as well as Daniel Cicala for his comments on the draft of this article. The author is grateful to the organizers of the Workshop on Information and Entropy in Biology held at the National Institute for Mathematical and Biological Synthesis (NIMBIOS) in Knoxville, TN as well as to NIMBIOS for its support in attending the workshop. Part of this project was completed during the author's visit to the Centre for Quantum Technologies (CQT) at the National University of Singapore (NUS) which  was supported by the NSF's East Asia and Pacific Summer Institutes Program (EAPSI) in partnership with the National Research Foundation of Singapore (NRF). 

\end{document}